\documentclass[acmsmall,nonacm,natbib=false]{acmart}
\usepackage[
  datamodel=acmdatamodel,
  style=acmnumeric,
  ]{biblatex}
\usepackage{amsthm}
\usepackage{thmtools,thm-restate}
\usepackage{amsmath}
\usepackage{lineno}
\usepackage{cleveref}
\usepackage{xcolor}
\usepackage{enumitem}

\newtheorem{theorem}{Theorem}
\newtheorem{lemma}{Lemma}
\newtheorem{algorithm}{Algorithm}
\newtheorem{definition}{Definition}
\newtheorem{example}{Example}

\DeclareMathOperator{\commit}{commit}
\DeclareMathOperator{\adopt}{adopt}
\newcommand{\tuple}[1]{\ensuremath{\langle #1\rangle}}

\newcommand{\remove}[1]{\ignorespaces}
\newcommand{\mayberemove}[1]{\ignorespaces}

\title{Consensus in the Unknown-Participation Message-Adversary Model}
\author{Giuliano Losa}
\email{giuliano@stellar.org}
\orcid{0000-0003-2341-7928}
\affiliation{%
  \institution{Stellar Development Foundation}
  \country{USA}
}
\author{Eli Gafni}
\email{eli@ucla.edu}
\affiliation{%
  \institution{University of California, Los Angeles} \country{USA}
}
\date{\today}

\begin{document}


\begin{abstract}
    We propose a new distributed-computing model, inspired by permissionless distributed systems such as Bitcoin and Ethereum, that allows studying permissionless consensus in a mathematically regular setting.
    Like in the sleepy model of Pass and Shi, we consider a synchronous, round-by-round message-passing system in which the set of online processors changes each round.
    Unlike the sleepy model, the set of processors may be infinite.
    Moreover, processors never fail; instead, an adversary can temporarily or permanently impersonate some processors.
    Finally, processors have access to a strong form of message-authentication that authenticates not only the sender of a message but also the round in which the message was sent.

    Assuming that, each round, the adversary impersonates less than 1/2 of the online processors, we present two consensus algorithms.
    The first ensures deterministic safety and constant latency in expectation, assuming a probabilistic leader-election oracle.
    The second ensures deterministic safety and deterministic liveness assuming irrevocable impersonation and eventually-stabilizing participation.

    The model is unrealistic in full generality.
    However, if we assume finitely many processes and that the set of faulty processes remains constant, the model coincides with a practically-motivated model: the static version of the sleepy model.
\end{abstract}

\maketitle
\section{Introduction}

A common feature of permissionless systems such as Bitcoin and Ethereum is that they are dynamically available.
This means that they continue working correctly even if a large fraction of the participants unexpectedly go offline or come back online, as long as a sufficient fraction (typically more than 1/2) of the online participants at any given time are well-behaved.

In contrast, traditional Byzantine fault-tolerant (BFT) consensus algorithms lose liveness in this setting because they rely on the availability of fixed-sized quorums.
Using proof-of-work, Bitcoin's Nakamoto consensus was the first dynamically-available consensus algorithm, and it is resilient to a failure ratio of 1/2 (which is optimal~\cite[Theorem 6.2]{lewis-pye_permissionless_2023}).
However, in Nakamoto consensus, the adversary can always be lucky and cause disagreement with some probability depending on a security parameter of the protocol.
Pass and Shi's work on the sleepy model~\cite{pass_sleepy_2017}, a formal model of dynamically-available systems, initiated a series of works on consensus algorithms that guarantee agreement probabilistically under a failure ratio of 1/2, like Nakamoto consensus, but without proof-of-work~\cite{pass_sleepy_2017,daian_snow_2019,badertscher_ouroboros_2018,damato_no_2022,goyal_instant_2021}.

More recently, Momose and Ren~\cite{momose_constant_2022} asked whether safety can be guaranteed deterministically in the dynamically-available setting.
Several works proposing consensus algorithms with deterministic safety under various assumptions have followed~\cite{malkhi_byzantine_2022,pu_safe_2022,pu_gorilla_2023,malkhi_towards_2023}, but the exact boundary of what can be achieved deterministically is still not clear.

In this paper, we propose to study deterministically-safe, dynamically-available consensus in a new model that we call the Iterated, Infinite Authenticated Byzantine model (IIAB).
The model is unrealistic but suitable for our goal of studying consensus in a mathematically regular model.
As an example of what we mean by regularity, the feasibility of deterministically-safe consensus in the IIAB model does not depend on whether the set of processors affected by faults is constant, monotonically increasing, or fluctuating.

The term iterated, as in iterated immediate snapshot (IIS)~\cite{borowsky_simple_1997}, conveys the repetition of synchronous, communication-closed rounds.
Indeed, the IIAB mode is a synchronous, round-by-round computing model where a potentially infinite set of processors are connected by point-to-point links forming a complete graph, and, each round, all messages sent are received by the end of the round.

However, each round $r$, an adversary partitions the processors into a set of offline processors, which do not send messages, and a finite and nonempty set of online processors $O_r$.
This is similar to the sleepy model, except that the set of processors may be infinite in our case.

More surprising is that the processors in the IIAB model never fail; instead, the adversary takes over their communication links as follows.
Each round, the adversary partitions the set of online processors $O_r$ into a nonempty set of well-behaved processors $W_r$ and a set of impersonated processors $F_r$, with $|F_r|< |W_r|$, that the adversary impersonates.
All processors faithfully follow their algorithm and never crash, but each round, the adversary can forge arbitrary messages on behalf of the impersonated processors.

Processors also have access to a strong form of message authentication that allows any processor to check whether a given message was sent by another given processor in a particular round.
Note that this strong form of authentication is unrealistic.
For example, in a system with a public signature scheme and a PKI, once an adversary obtains a private key $k$, the adversary can sign a message with $k$ and claim that it was sent in a previous round before the adversary obtained the key.
This is not possible in the IIAB model.

The IIAB model may seem unrealistic, but it subsumes some models used in previous work on permissionless systems.
For example, the static version of the sleepy model of Pass and Shi~\cite{pass_sleepy_2017} coincides with the IIAB model under the assumption that the set of processes is finite and that, while participation is dynamic, the set of impersonated processors is fixed.
If both online and impersonated processors are fixed and additionally known to the processors, we obtain the traditional, synchronous authenticated Byzantine agreement model.
Thus, solving consensus in the IIAB model solves it in those models too.

Solving consensus in the IIAB model under the minority-corruption assumption is non-trivial.
There are two well-known approaches to solving consensus in synchronous systems with message authentication: the Dolev-Strong algorithm~\cite{dolev_authenticated_1983} and quorum-based approaches, as in the phase king algorithm~\cite{berman_towards_1989}.
Unfortunately, none of those approaches seems to work in the IIAB model.

First, our assumptions about participation make it possible that no well-behaved processor be online in more than one round.
Thus, it is useless for processors to keep a local state.
This rules using the Dolev-Strong algorithm~\cite{dolev_authenticated_1983} because it crucially relies on local state.

Second, due to the 1/2 failure ratio, we can only rely on majority quorums (and no bigger), and it seems that majority quorums cannot prevent disagreement in the IIAB model:
In a given round, even assuming that each processor broadcasts a unique message, it is possible that an online, well-behaved processor $p$ receives a message $m$ from a strict majority of the processors it hears of, while another online, well-behaved processor $p'$ receives a message $m'\neq m$ from a strict majority of the processor it hears of.
For example:
\begin{example}
    \label{ex:contradicting-majs}
    Consider a round $r$ in which the online processors are $\{p_1,p_2,p_3\}$, and suppose that $F_r=\{p_1\}$ is impersonated while $W_r=\{p_2,p_3\}$ is well-behaved.
    Suppose that $p_2$ broadcasts $v$ while $p_3$ broadcasts $v'\neq v$.
    If the adversary sends $v$ to $p_2$ on behalf of $p_1$ and $v'$ to $p_3$ on behalf of $p_1$, then $p_2$ receives $v$ from a strict majority, while $p_3$ receives $v'$ from a strict majority.
\end{example}

\subsection{Summary of the Contributions}

The paper makes the following main contributions.
\begin{enumerate}
    \item We propose the IIAB model. The IIAB model is an abstract model inspired by permissionless systems with fluctuating participation, and it allows one to study consensus in a mathematically regular setting.
        The IIAB model is a generalization of the static version of the sleepy model of Pass and Shi~\cite{pass_sleepy_2017}, which coincides with the IIAB model under the assumption that the set of processes is finite and that, while participation is dynamic, failures are static.
    \item As usual for consensus, the main difficulty lies in solving the commit-adopt task~\cite{gafni_round-by-round_1998} (a graded agreement~\cite{chen_algorand_2019} with two grades).
         In~\Cref{sec:ca}, we propose a commit-adopt algorithm that generalizes the algorithm of Gafni and Zikas~\cite{gafni_synchronyasynchrony_2023}.
    \item We present a probabilistic consensus algorithm relying on a probabilistic leader-election oracle.
        Under reasonable assumptions about the oracle (implementable using verifiable random functions~\cite{micali_verifiable_1999} in the static sleepy model), this algorithm terminates in 20 rounds in expectation.
    \item We present a deterministic consensus algorithm assuming that participation eventually stabilizes, i.e.\ that it remains constant after some unknown round $R$, and that the adversary is growing (i.e., for every round $r$, $F_r\subseteq F_{r+1}$).
            The algorithm terminates in $O(\left|O_R\right|)$ rounds after round $R$.
\end{enumerate}

In supplemental material\cite{losa_nano-odynamic-participation-supplemental_2023}, we also provide TLA+ specification of the algorithms and a mechanically-checked proof in Isabelle/HOL of the key lemma behind the correctness of the commit-adopt algorithm.

\subsection{Technical Outline}

To solve consensus in the IIAB model, we proceed in two steps.
First, we simulate a model that is easier to work with, which we call the no-equivocation IIAB model (no-equivocation model for short), and then we present consensus algorithms in the no-equivocation model.

Compared to the IIAB model, the no-equivocation model places additional restrictions on what the adversary can do: in any given round, the adversary is restricted to sending the same message on all channels on which it sends a message (it cannot equivocate by sending different messages to different processors);
moreover, if the adversary sends a message from processor $p$ to processor $p_1$ but not to processor $p_2$, then the adversary must send a special failure notification, noted $\lambda$, from $p$ to~$p_2$.

Crucially, the properties of the no-equivocation model imply that it is no longer possible for two processors to witness two conflicting strict majorities in the same round, which makes it much easier to devise algorithms based on majority quorums.

We present the no-equivocation model in~\Cref{sec:non-eq} and we give the simulation algorithm, which simulates the no-equivocation model in the IIAB model, in~\Cref{sec:simulation}.

Next, in~\Cref{sec:consensus}, we turn our attention to solving consensus in the no-equivocation model.
We follow the classic pattern, presented in~\Cref{sec:alternation}, consisting in executing an alternating sequence of conciliator and commit-adopt~\cite{gafni_round-by-round_1998} phases.
This pattern, illustrated in~\Cref{fig:alternating-seq}, makes designing consensus algorithms easier by separating concerns.
Each conciliator phase must bring processors into agreement when conditions are favorable (e.g.\ when all processors happen to agree on a good leader, or when the system has stabilized in some way), but is otherwise allowed to fail to bring processors into agreement.
Each commit-adopt phase must detect a preexisting agreement, if there is one, and produce a consensus decision.

In~\Cref{sec:ca}, we present a commit-adopt algorithm that takes 2 no-equivocation rounds.
This algorithm generalizes the algorithm of Gafni and Zikas~\cite{gafni_synchronyasynchrony_2023} to the case of dynamic participation.

In~\Cref{sec:proba-consensus}, we present a probabilistic conciliator which guarantees agreement with probability~1/2.
This conciliator relies on a probabilistic leader-election oracle that, each round and with probability 1/2, provides the identity of a unique well-behaved leader to all well-behaved processors.
In the static sleepy model, this oracle is implementable with verifiable random functions (VRF)~\cite{micali_verifiable_1999}.
Using the probabilistic conciliator, we obtain a consensus algorithm that terminates in 20 rounds in expectation.

Finally, in~\Cref{sec:dolev-strong}, we present a deterministic conciliator for the case of constant participation inspired by the Dolev-Strong algorithm.
This conciliator algorithm is parameterized by a natural number $N>0$, it runs for $N+1$ rounds, and it guarantees agreement if participation is constant and less than $N$ participants are impersonated during the first $N$ rounds (i.e.\ $\left|\bigcup_{0<r\leq N}F_r\right|<N$).

Using this deterministic conciliator, we obtain a consensus algorithm that guarantees deterministic safety and deterministic liveness under a growing adversary (i.e., for every round $r$, $F_r\subseteq F_{r+1}$) when there is an unknown round~$R$, called the stabilization round, after which participation remains constant.
In this case, the algorithm terminates in $O(\left|O_R\right|)$ rounds after round $R$.

\subsection{Previous version of this paper}

Compared to the previous version of the paper, we clarify the exposition, we simplify the proof of the commit-adopt algorithm, and we correct a number of errors:
\begin{itemize}
        \item We remove the erroneous claim that the IIAB model subsumes the sleepy model with dynamic corruptions (the IIAB model only subsumes the static sleepy model).
        \item We correct the deterministic consensus algorithm of~\Cref{sec:dolev-strong}, as well as the conditions under which it guarantees deterministic termination.
\end{itemize}

\section{The IIAB model and the no-equivocation IIAB model}
\label{sec:model}

\subsection{The IIAB model}

We consider a potentially infinite set $\mathcal{P}$ of processors running a distributed algorithm in a synchronous, message-passing system with point-to-point links between every pair of processors.

\subsubsection{Online/offline and well-behaved/impersonated sets.}
Processors execute a sequence of synchronous communication rounds numbered $1,2,3,\dots$.
Each round $r$, a nonempty, finite set $O_r$ of processors is online while the set $\mathcal{P}\setminus O_r$ is offline;
moreover, the set of online processors $O_r$ is partitioned into a nonempty set of well-behaved processors $W_r$ and a set of impersonated processors $F_r$ such that there are fewer impersonated online processors than well-behaved online processors (i.e.\ $|F_r| < |W_r|$).
The sets $O_r$, $W_r$, and $F_r$ are fixed for all rounds $r$ but unknown to the processors.

We think of impersonations as being performed by an adversary.
We say that the adversary is growing when $F_r\subseteq F_{r+1}$ for every round $r$.

\subsubsection{Processor and adversary behavior.}
Operationally, in round $r=1$, each processor $p$ (whether online or offline) first receives an external input $in_p$.
Then, each round $r$ consists of a send phase followed by a receive phase.
In the send phase, each online processor sends a set of messages determined, as prescribed by the algorithm, by the messages it received in the previous round and the output of queries it makes to some oracles, if $r>1$, or solely by its external input if $r=1$.
However, the adversary removes all outgoing messages of online impersonated processors (the set $F_r$) and injects arbitrary messages, subject to the restrictions below, in the outgoing communication links of impersonated processors.
Finally, offline processors do not send messages.

In the receive phase of each round $r$, every processor (whether online or not) receives all the messages sent to it during the send phase of the same round $r$.
Moreover, each processor may optionally produce an irrevocable output, as prescribed by the algorithm, that depends on the set of messages it receives in the round and on the output of any oracle calls it may make.

\subsubsection{Message authentication}

Messages may be arbitrary, except that there is a type of signed messages with special restrictions.
Signed messages are of the form $\langle p,r,m\rangle$, where $p$ is the signer of the signed message, $r$ the round of the signed message, and $m$ the message content of the signed message.
Moreover, in a each round $r$, the adversary can only send a signed message $\langle p,r',m\rangle$ if $r'=r$ and it impersonates the signer $p$ of the message in round $r$, unless $\langle p,r',m\rangle$ has been sent in a previous round.
Formally, if a signed message $\langle p,r,m\rangle$ appears (either on its own or included in another message) on a communication link $(p_1,p_2)$ in round $r'$, then:
\begin{itemize}
        \item If $p_1\in W_{r'}$, then either $p=p_1$ and $r'=r$, or $r<r'$ and $p_1$ received $\langle p,r,m\rangle$ (either directly or as part of another message) in a round $r''<r'$.
        \item If $p_1\in F_{r'}$, then either $p\in F_r$ and $r'=r$, or $r< r'$ and the signed message $\langle p,r,m\rangle$ was sent on a communication link in a round $r''<r'$.
\end{itemize}

One consequence of the message restrictions above is that, if the adversary impersonates a processor $p$ in round $r$ but not in round $r'\neq r$, then the adversary cannot credibly claim that $p$ sent a message $m$ in round $r'$ when $p$ actually did not send that message.
This in contract with a system based on public-key authentication, where, once the private key of a processor $p$ is compromised, the adversary can credibly claim that $p$ sent messages in the past when $p$ in fact did not.

\subsubsection{The leader-election oracle}
When indicated, processors have access to a probabilistic leader-election oracle offering a local interface to each processor.

Each round $r$, the leader-election oracle can be queried by each processor $p$ and returns the identity of a processor $l_p^r$, called a leader, with the guarantee that, with probability 1/2, there is a well-behaved processor $l_r$ such that every processor $p$ obtains the same leader $l_p^r=l_r$.

The leader-election oracle is an abstraction of using verifiable random functions~\cite{micali_verifiable_1999} (VRF) to elect a leader.
VRFs can be used to elect a leader among a finite set of processors by electing the processor with largest VRF output; if a minority of processors selectively reveal their outputs to others, all well-behaved processors nevertheless agree on the same leader with probability 1/2.

\subsubsection{Composing Algorithms}

In later sections, we construct algorithms by sequentially composing sub-algorithms.
In such cases, the sub-algorithms all use a fixed number of rounds, and thus composition only amounts to executing the algorithms in a sequence, feeding the outputs of an algorithm as inputs to the next.

\subsubsection{Useful terminology.}
We say that $p$ hears of $q$ in round $r$ when $p$ receives some message on the link $(p,q)$ in the receive phase of round $r$.
We say that $p$ receives a message $m$ from a strict majority in round $r$ when $p$ receives $m$ from strictly more than half of the processors it hears of in round $r$.

\subsubsection{A note about strict majorities}

It is important to note that a set of processors $P$ may be a strict majority of the processors that a processor $p$ hears of in round $r$ without $P$  necessarily being a strict majority of the online processors.
\begin{example}
    Take 4 processors $p_1$ to $p_4$ with $W_r=\{p_1,p_2,p_3\}$ and $F_r=\{p_4\}$ (so $O_r=\{p_1,p_2,p_3,p_4\}$).
    Suppose that $p_1$ hears of $\{p_1,p_2,p_3\}$ only in round $r$.
    Then $P=\{p_1,p_2\}$ is a strict majority of the processors that $p_1$ hears of in round $r$, but it is not a strict majority of the online processors $\{p_1,p_2,p_3,p_4\}$.
\end{example}
Similarly, as shown in~\Cref{ex:contradicting-majs} in the introduction, even supposing that each processor broadcasts a unique message in some round $r$, it is possible that two processors each hear from two different messages from a strict majority.
However, a useful property on which we can rely is that if $p$ receives a message $m$ from a strict majority in a round $r$, then at least one well-behaved processor sent $m$ in round $r$.

\subsection{Relationship to other models}
\label{sec:other-models}

Special cases of the IIAB model correspond to existing models, if we assume there are finitely many processors:
\begin{enumerate}
    \item If $F_r$ remains constant, then we obtain the static version of the Sleepy model of Pass and Shi~\cite{pass_sleepy_2017}.
    \item If we assume that both $W_r$ and $F_r$ remain constant (that is, $F_r=F_1$, $W_r=W_1$, and thus $O_r=O_1$ for every $r$), and that both $W_r$ and $F_r$ are known to the processors, then we obtain the traditional synchronous, authenticated Byzantine agreement model.
    \item If $O_r$ remains constant, $|F_r|\leq 1$, and the adversary is restricted to just dropping messages (and not forging messages), then we obtain the model of Santoro and Widmayer~\cite{santoro_time_1989}
\end{enumerate}

Note that the sleepy model with dynamic corruptions has no equivalent in the IIAB model.
This is because the message restrictions we place on signed messages are not implementable with public-key cryptography: once the adversary obtains a key, it can forge messages and undetectably claim they were sent in previous rounds.
Our message restrictions prevents this.

\subsection{Solvability of consensus in the IIAB model}

Santoro and Widmayer~\cite{santoro_time_1989} show that deterministic consensus (with both deterministic safety and deterministic liveness) is impossible even if we restrict the adversary a) to impersonating only one processor each round and b) to only dropping but never fabricating messages.
Thus, deterministic consensus is impossible in the IIAB model without additional assumptions.

\subsection{The no-equivocation IIAB Model}
\label{sec:non-eq}

We now place additional restrictions on the IIAB model to obtain what we call the no-equivocation IIAB model (no-equivocation model, for short).

First, the no-equivocation model requires processors to only broadcast a unique message per round (when, in the IIAB model, processors are free to send any number of messages on each of their outgoing links).

Second, the no-equivocation model places additional restrictions on the adversary: each round $r$, for each impersonated processor $p$, the adversary picks a single message $m$ and sends $m$ on some of $p$'s outgoing channels (possibly none); moreover, should the adversary send $m$ from $p$ to some processors but not others, then the adversary must send the special failure notification $\lambda$ to all the processors to which it does not send $m$; finally, the adversary can also send $\lambda$ on some of $p$'s outgoing channels and remain silent on others.

Thus, for each round $r$ and each processor $p$, all the following combinations are possible:
\begin{enumerate}
    \item No processor hears of $p$ (possible when $p$ is offline, or $p \in F_r$ and the adversary just drops all of $p$'s messages), or
    \item All processors receive the same message $m$ that $p$ actually sent (possible when $p$ is online well-behaved or the adversary acts as if $p$ were well-behaved), or
    \item There is a message $m$ (possibly different from the messages that $p$ actually broadcast) such that all processors receive a message from $p$, but some processors receive $m$ for $p$ and others receive $\lambda$ (possible when $p$ is impersonated), or
    \item Some processors receive $\lambda$ for $p$ and some do not hear at all of $p$ (possible when $p$ is impersonated).
\end{enumerate}

\section{Simulating the no-equivocation model}
\label{sec:simulation}

In this Section, we show how to implement, in the IIAB model, the no-equivocation model defined in~\Cref{sec:non-eq}.

\begin{algorithm}[Simulation of no-equivocation model]
    We simulate one round of the no-equivocation model using two consecutive rounds of the IIAB model.
    \label{algo:non-eq}
    \begin{enumerate}[label=(\alph*)]
        \item In the first round, each online processor $p$ does the following: for each message $m$ that $p$ must simulate the broadcasting of, $p$ broadcasts the signed message $\tuple{p,1,m}$.
        \item In the second round, for each process $p'$ and each signed message $\tuple{p',1,m}$ received from $p'$, each online processor $p$ forwards $\tuple{p',1,m}$ to all (thereby relaying that $p'$ sent $m$ in round 1) by broadcasting $\tuple{p,2,\tuple{p',1,m}}$.
            We say that $\tuple{p,2,\tuple{p',1,m}}$ is a claim from $p$ that $p'$ sent message $m$ in round 1.
        \item At the end of the second round, each processor $p$ simulates receiving the following messages.
            For each processor $p'$ such that $p$ receives a claim that $p'$ sent a message $m$ in round 1, $p$ simulates the delivery of the following message:
            \begin{enumerate}[label=(\roman*)]
                \item If $p$ receives claims that $p'$ sent $m$ in round 1 from a strict majority in round 2 and $p$ does not receive any claim that $p'$ sent any other message $m'\neq m$ in round 1, then $p$ simulates receiving $m$ from $p'$.\label{item:no-eq-sim-i}
                \item Else $p$ simulates receiving $\lambda$ from $p'$\label{item:no-eq-sim-ii}.
            \end{enumerate}
    \end{enumerate}
\end{algorithm}

\begin{theorem}
    \Cref{algo:non-eq} correctly simulates the no-equivocation model.
    \label{thm:non-eq-correct}
\end{theorem}
\begin{proof}
    By the definition of the no-equivocation model, it suffices to show that:
    \begin{enumerate}
        \item If $p_1$ simulates receiving $m\neq \lambda$ from $q$ then each processor $p_2$ simulates receiving either $m$ or $\lambda$ for $q$.\label{pf:non-eq-correct-case-1}
        \item If $q$ is a well-behaved processor in round 1, then every processor simulates receiving $q$'s input from~$q$.\label{pf:non-eq-correct-case-2}
    \end{enumerate}
    We start with~\Cref{pf:non-eq-correct-case-1}.
    Suppose that $p_1$ simulates receiving $m\neq\lambda$ for processor $q$ and consider a processor $p_2$.
    Since $p_1$ simulates delivering $m$ from $q$, $p_1$ receives claims that $q$ send message $m$ in round 1 from a strict majority of the processors it hears of in round 2.
    Therefore, at least one well-behaved processor made the claim, and therefore $p_2$ also receives the claim and, by~\Cref{item:no-eq-sim-i,item:no-eq-sim-ii} above, $p_2$ simulates receiving either $m$ or $\lambda$ from $q$.

    Next, we turn to~\Cref{pf:non-eq-correct-case-2}.
    Consider a processor $q$ that is well-behaved in round 1 with input message $m$.
    Since $q$ is well-behaved, every processor receives $m$ from $q$ by the end of round 1.
    Then, in round 2, all the processors that are well-behaved in round 2 broadcast a claim that $q$ sent $m$ in round 1.
    Moreover, by the message restrictions of the IIAB model, no processor can broadcast a claim that $q$ sent a message different from $m$ in round 1.
    Thus, at the end of round 2, all processors simulate receiving $m$ from $q$.
\end{proof}

\section{Solving Consensus in the no-equivocation model}
\label{sec:consensus}

We start by defining the consensus problem.

\begin{definition}
    \label{def:consensus}
    In the consensus problem, every processor receives an external input and may produce an output such that:
    \begin{itemize}
        \item Agreement: No two processors produce different outputs.
        \item Validity: If all processors have the same input $v$, no processor outputs $v'\neq v$.
    \end{itemize}
    Moreover, one of the following termination properties must additionally be satisfied.
    \begin{itemize}
        \item Eventual termination: eventually, all processors output a value.
        \item Probabilistic termination: with probability 1, all processors eventually output.
    \end{itemize}
\end{definition}
Note that probabilistic termination allows some executions in which some processors never output.
However, those executions must have probability zero.
Moreover, note that agreement and validity must be guaranteed in all executions (so, an algorithm like Nakamoto consensus does not satisfy our definition of consensus).

\subsection{Solving consensus by alternating conciliators and commit-adopt}
\label{sec:alternation}

To solve consensus, we use the classic approach consisting in alternating between a conciliator phase and a commit-adopt~\cite{gafni_round-by-round_1998} phase (commit-adopt is a graded agreement~\cite{chen_algorand_2019} with two grades), starting with a conciliator.
\Cref{fig:alternating-seq} visually depicts the approach.

In a conciliator phase, processors try to reach agreement on a value, but they might fail to do so and end up in disagreement.
In a commit-adopt phase, processors commit to a decision if they successfully reached agreement in the previous conciliator phase.
As we will see, the exact properties that conciliator and commit-adopt phases must satisfy guarantee the validity and agreement properties of consensus.
However, the termination property that is achieved depends on the conditions under which processors are guaranteed to successfully reach agreement in the conciliator phase; this depends on the algorithm and on additional assumptions (like the availability of the leader-election oracle or the bounded-failures assumption).

\begin{figure}[ht]
        \caption{Infinite alternating sequence of conciliator and commit-adopt phases. Horizontal arrows represent processors locally taking their output from one phase and using it as input to the next phase; vertical arrows represent processors possibly producing a consensus decision. Processors keep executing and move from one phase to the next regardless of whether they have produced an output or not.}
        \begin{center}
            \includegraphics[width=\textwidth]{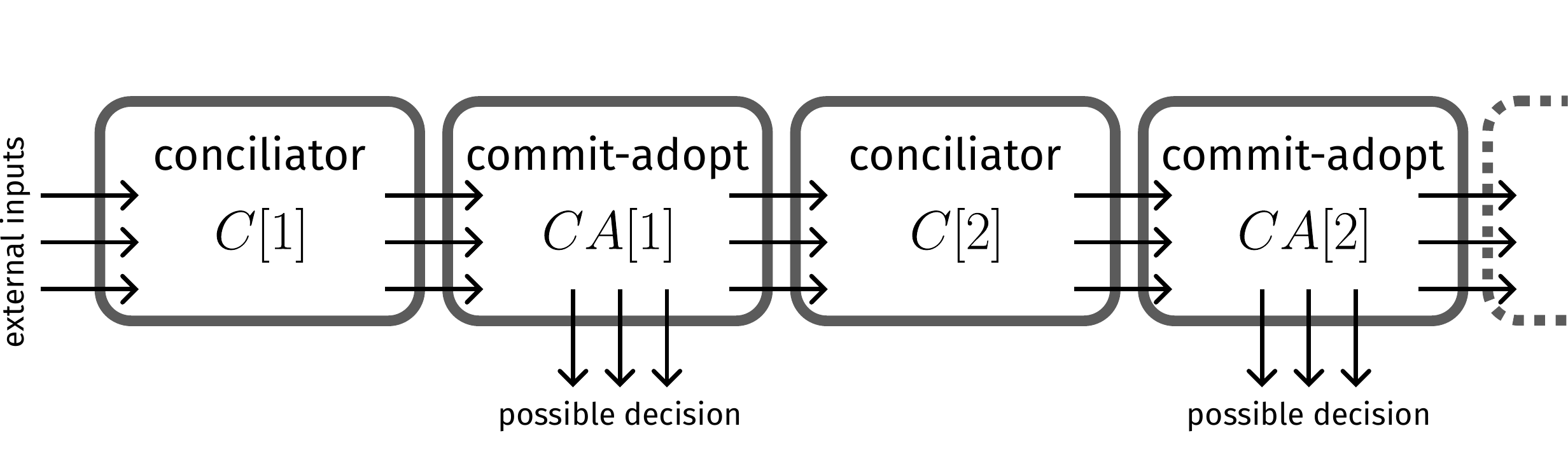}
        \end{center}
        \label{fig:alternating-seq}
\end{figure}

Formally, a commit-adopt algorithm must obey the following specification.
\begin{definition}
    \label{def:ca}
    In the commit-adopt task, each processor starts with an input value and must produce an output of the form $commit(v)$ or $adopt(v)$, for some value $v$, subject to three conditions:
    \begin{itemize}
        \item[Agreement] If a processor outputs $commit(v)$ then every processor must output $commit(v)$ or $adopt(v)$.
        \item[Validity] If all processors have the same input $v$, then every processor must output $commit(v)$.
        \item[Termination] There is a round $R$ such that all processors output in round $R$.
    \end{itemize}
\end{definition}

A conciliator algorithm must obey the following specification.
\begin{definition}
    \label{def:conciliator}
    In the conciliator task, each processor starts with an input value and must produce an output subject to the two conditions below.
    \begin{itemize}
        \item[Validity] If there is a value $v$ such that all processors have input $v$, then all processors output $v$.
        \item[Termination] There is a round $R$ such that all processors output in round $R$.
    \end{itemize}
\end{definition}
Note that the conciliator algorithms that we present in later sections additionally satisfy progress properties, described in~\Cref{sec:proba-consensus,sec:dolev-strong}, that require the conciliator to produce agreement with some probability or, depending on some parameter of the algorithm, deterministically.

We now give a generic algorithm that we define in terms of two arbitrary conciliator and commit-adopt algorithms.
\begin{algorithm}[Generic consensus algorithm]
    \label{algo:gen-consensus}
    We use an infinite sequence $C[1,2,\dots]$ of conciliator instances and an infinite sequence $CA[1,2,\dots]$ of commit-adopt instances.
    In the generic consensus algorithm, each round $r$, each online processor $p$ does the following.
    \begin{enumerate}[label=(\alph*)]
        \item If $r=1$, invoke $C[1]$ with $p$'s external input.
        \item Upon obtaining an output $o$ from the conciliator $C[n]$, for some $n$, invoke $CA[n]$ with $o$ as input.
        \item Upon obtaining an output $o$ from the commit-adopt $CA[n]$, for some $n$:
            \begin{itemize}
                \item If $o=\commit(v)$ for some $v$, then output $v$ as consensus decision.
                \item Invoke $C[n+1]$ with the value $v$ such that $o=\commit(v)$ or $o=\adopt(v)$.
            \end{itemize}
    \end{enumerate}
\end{algorithm}

Informally, let us reiterate that the goal of the conciliator phase is to bring processors into agreement when conditions are favorable (e.g.\ when all processors happen to agree on a good leader, or when the system has stabilized in some way), while the goal of the commit-adopt phase is to detect a pre-existing agreement and produce a consensus decision.

\begin{theorem}
    \Cref{algo:gen-consensus} satisfies the validity and agreement properties of consensus.
\end{theorem}
\begin{proof}
    Let us start with the validity property.
    Suppose that all processors receive the same external input $v$.
    The, by property of conciliators, all processors output $v$ in the first conciliator $C[1]$.
    Thus, by property of commit-adopt, all processors decide $v$ in the first commit-adopt $CA[1]$, and we are done.

    For the agreement property, note that, by the specification of conciliators and commit-adopt, once a processor outputs $\commit(v)$, for some $v$, from a commit-adopt instance $CA[n]$, then all processors input $v$ to the next conciliator $C[n+1]$; thus all processors output $v$ in $C[n+1]$, which means all processors commit $v$ in the commit-adopt $C[n+1]$ and all decide $v$.
\end{proof}

We will discuss how we achieve the termination properties of consensus in later sections, when we present the respective conciliator algorithms (see~\Cref{thm:expected-rounds,thm:det-consensus}).

\subsection{Implementing commit-adopt in the no-equivocation model}
\label{sec:ca}

\begin{algorithm}[Commit-adopt]
    \label{algo:ca}
    The algorithm takes 2 rounds of the no-equivocation model.
    \begin{enumerate}[label=(\alph*)]
        \item In the first round, each online processor broadcasts its input $in_p$.
        \item In the second round, each online processor $p$ broadcasts propose-commit$(v)$ if it hears $v$ from a strict majority; otherwise, $p$ broadcasts no-commit.
        \item At the end of the second round, each processor $p$ determines its output as follows:
            \begin{enumerate}[label=(\roman*)]
                \item If $p$ received propose-commit$(v)$ from a strict majority, for some $v$, then it outputs $commit(v)$.\label{ca:output-case-1}
                \item Else, if $p$ received a value $v$ from at least one processor and from strictly more processors than any other value, then $p$ outputs $adopt(v)$.\label{ca:output-case-2}
                \item Else, $p$ outputs adopt$(in_p)$.\label{ca:output-case-3}
            \end{enumerate}
    \end{enumerate}
\end{algorithm}

Although the algorithm is simple, it is no straightforward to prove it correct.
We begin with some useful properties of the no-equivocation model.

\subsubsection{Key Properties of the no-equivocation model}

In this Section we present~\Cref{thm2,thm3}, which capture key properties of the no-equivocation model that we will use in the next subsection to prove the algorithm correct.

First, in~\Cref{thm2}, we show that no two processors can receive two different messages from the same processor in the same round.
Thus, despite unknown participation, we can rely on a property similar to the quorum-intersection property that many classic BFT algorithms rely on.

Second, in~\Cref{thm3}, we show that if a processor receives $m$ from a strict majority in round $r$ and if no processor broadcasts $m'$ in round $r$, then every processor receives $m$ strictly more often than $m'$ in round $r$.
We use this property to ensure that, should a processor act after hearing from a unanimous strict majority, other processors have an indication that this might have happened.
This is similar to the property, used in classic BFT algorithms, that if a processor hears from a unanimous super-majority, i.e.\ $2f+1$ out of $3f+1$, then all processors eventually hear from $f+1$ correct processors out of those $2f+1$.

Let us start with some notational conventions.
When a round $r$ is clear form the context, let $H_{p}$ be the set of processors that $p$ hears of in $r$ and let $H_{p'}$ be the set of processors that a processor $p'$ hears of in $r$.

We first need a simple lemma that we use to prove~\Cref{thm2,thm3}.

\begin{lemma}
    \label{l1}
    Consider a round $r$.
    Suppose that a processor $p$ receives a message $m$ from each member $q$ of a set $P_m$ of processors.
    Then, for every processor $p'$, we have:
    \begin{enumerate}
        \item $P_m\subseteq H_{p}\cap H_{p'}$, and
        \item if $P_m$ is a strict majority of $H_{p}$, then $P_m$ is a strict majority of $H_{p}\cap H_{p'}$.
    \end{enumerate}
\end{lemma}
\begin{proof}
    We start with Item 1.
    First note that $P_m\subseteq H_{p}$ by definition.
    By definition of the no-equivocation model, for every $q\in P_m$, if $p'$ does not receive $m$ then $p'$ receives $\lambda$ for $q$.
    Hence, $p'$ hears of every processor $q$ such that $p$ receives $m$ from $q$, i.e.\ $P_m\subseteq H_{p'}$.
    Thus, we have $P_m\subseteq H_{p}\cap H_{p'}$.

    For Item 2, note that we have just established that $P_m\subseteq H_{p}\cap H_{p'}$.
    Therefore, if $P_m$ is a strict majority of $H_{p}$, by simple properties of sets, we get that $P_m$ is a strict majority among $H_{p}\cap H_{p'}$.
\end{proof}

\begin{theorem}
    \label{thm2}
    Consider a round $r$.
    For every two different value $m_1$ and $m_2$, if a processor $p$ receives $m_1$ from a strict majority, then no processor $p'$ receives $m_2$ from a strict majority.
\end{theorem}
\begin{proof}
    Suppose by contradiction that $p$ receives $m_1$ from a strict majority $Q_1$ of the processors it hears of and that $p'$ receives $m_2$ from a strict majority $Q_2$ of the processors it hears of.
    Because the adversary cannot equivocate in the no-equivocation model, we have that $Q_1$ and $Q_2$ are disjoint.

    However, by~\Cref{l1}.2, $Q_1$ consists of a strict majority of $H_{p}\cap H_{p'}$.
    Similarly, $Q_2$ consists of a strict majority of $H_{p}\cap H_{p'}$.
    This contradicts the fact that $Q_1$ and $Q_2$ are disjoint.
\end{proof}

\begin{theorem}
    \label{thm3}
    Consider a round $r$.
    Suppose that a processor $p$ receives $m_1$ from a strict majority.
    Moreover, suppose that no processor broadcasts $m_2\neq m_1$ in round $r$.
    Then, for every processor $p'$, if $m_2\neq \lambda$, then $p'$ receives $m_2$ from strictly less processors than it receives $m_1$ from.
\end{theorem}
\begin{proof}

    First, note that $m_1\neq \lambda$ because at least one well-behaved processor sends $m_1$.

    Let $H=H_p\cap H_p'$ be the set of processors that both $p$ and $p'$ hear from.
    Let $F=F_r\cap H$ be the set of impersonated processors that both $p$ and $p'$ hear from.
    Note that, by definition of the model, $W_r\subseteq H$; thus $2|F|< |H|$ and $F\cap H$ consists of a strict minority of $H$.

    Let $M_1$ be the set of processors from which $p$ receives $m_1$.
    By~\Cref{l1}, we have $M_1\subseteq H$ and $M_1$ consists of a strict majority of $H$.
    With the fact that $F$ is a strict minority of $H$, we have $|M_1| > |F|$.
    Therefore, by simple properties of finite sets, we have 
    \begin{equation}
        \label{eq:sets-card}
        |M_1\setminus F| > |F\setminus M_1|
    \end{equation}

    Let $M_1'$ be the set of processors from which $p'$ receives $m_1$.
    Let $M_2'$ be the set of processors from which $p'$ receives $m_2$.
    We must show that $|M_2'|< |M_1'|$.

    Since each well-behaved processor broadcasts a unique message, we have $M_1\setminus F\subseteq M_1'$.
    Moreover, since no processor broadcasts $m_2\notin\{m_1,\lambda\}$ and since the adversary cannot equivocate, we have $M_2'\subseteq F\setminus M_1$.
    Using~\Cref{eq:sets-card}, we easily get that $|M_1'| > |M_2'|$ and we are done.
\end{proof}

In supplemental material~\cite{losa_nano-odynamic-participation-supplemental_2023}, we provide a mechanically-checked proof of \Cref{thm3} in Isabelle/HOL.

We are now ready to prove the commit-adopt algorithm correct.

\subsubsection{Proof of correctness of the commit-adopt algorithm}

\begin{theorem}
    \Cref{algo:ca} implements commit-adopt.
    \label{thm:ca-correct}
\end{theorem}
\begin{proof}
    First, we show the safety property of commit-adopt.
    Assume by contradiction that a processor $p$ outputs $commit(v)$ and a processor $p'$ outputs $commit(v')$ or $adopt(v')$ and $v\neq v'$.

    Suppose that $p'$ outputs according to~\Cref{ca:output-case-1} in~\Cref{algo:ca}.
    Then $p'$ receives $\textit{propose-commit}(v')$ from a strict majority.
    Moreover, since $p$ outputs $commit(v)$, $p$ receives \textit{propose-commit}$(v)$ from a strict majority.
    With $v\neq v'$, this contradicts~\Cref{thm2}.

    Now suppose that $p'$ does not output according to~\Cref{ca:output-case-1} in~\Cref{algo:ca}.
    Note that, if a processor broadcasts $\textit{propose-commit}(v)$, then, by~\Cref{thm2}, no processor broadcasts $\textit{propose-commit}(w)$ for $w\neq v$.
    Moreover, since $p$ outputs $commit(v)$, $p$ receives $\textit{propose-commit}(v)$ from a strict majority in round 2.
    Therefore, we satisfy the premises of~\Cref{thm3} with $m_1=\textit{propose-commit}(v)$ and $m_2=\textit{propose-commit}(v')$.
    Thus, by~\Cref{thm3}, $p'$ receives $\textit{propose-commit}(v')$ strictly less often than $\textit{propose-commit}(v)$.
    Thus $p'$ outputs $adopt(v)$ according to~\Cref{ca:output-case-2} in~\Cref{algo:ca}.
    With $v\neq v'$, this contradicts the assumption that $p'$ outputs $commit(v')$ or $adopt(v')$.

    Next, we show the validity property of commit-adopt.
    Suppose that all processors have input~$v$.
    Then, every processor receives $v$ from a strict majority in the first round.
    Thus, every processor broadcasts $\textit{propose-commit}(v)$ in round 2.
    Similarly, every processor receives $\textit{propose-commit}(v)$ from a strict majority in the second round.
    Thus, every processor outputs $commit(v)$ and we are done.
\end{proof}

Now that have presented the commit-adopt algorithm, it remains to present the conciliators that we use in the two consensus algorithms, and how they achieve the desired liveness properties.

\subsection{Consensus with probabilistic termination}
\label{sec:proba-consensus}

In this Section, we present a conciliator algorithm which, paired with the commit-adopt algorithm of~\Cref{sec:ca} in the generic consensus algorithm of~\Cref{sec:alternation}, allows solving consensus with deterministic safety and probabilistic termination in 20 rounds in expectation.

In addition to the properties of~\Cref{def:conciliator}, the conciliator algorithm we present in this section satisfies the following probabilistic agreement property:
\begin{definition}[Probabilistic agreement]
    \label{def:proba-agreement}
    With probability 1/2, every processor outputs the same value.
\end{definition}

The idea of the probabilistic conciliator is to rely on a good, unanimously agreed-upon leader (given with probability 1/2 by the leader-election oracle) to bring all processors into agreement.
However, blindly trusting the leader would put the validity property of conciliators at risk when we are unlucky and the leader is impersonated by the adversary.
Instead, we proceed as follows.
\begin{algorithm}[Probabilistic Conciliator]
    \label{algo:proba-conciliator}
    The algorithm takes 3 rounds of the no-equivocation model.
    In the first two rounds, participants do commit-adopt according to~\Cref{algo:ca}.
    Then, in the third round, each participant $p$ broadcasts its commit-adopt output and then outputs as follows:
    \begin{enumerate}
        \item If $p$ received $commit(v)$ from a strict majority, then $p$ outputs $v$.\label{algo:proba-output-1}
        \item Else, if $p$ received $commit(v)$ or $adopt(v)$ from its leader $l_p^3$, then $p$ outputs $v$.\label{algo:proba-output-2}
        \item Else $p$ outputs its initial, external input.\label{algo:proba-output-3}
    \end{enumerate}
\end{algorithm}

\Cref{fig:proba-cons} depicts the generic consensus algorithm instantiated with~\Cref{algo:proba-conciliator} as conciliator.
\begin{figure}[ht]
        \caption{Three processors $p_1$, $p_2$, and $p_3$ executing an infinite alternating sequence of conciliator and commit-adopt phases as per~\Cref{algo:gen-consensus}. Each conciliator phase consists of three rounds as per~\Cref{algo:proba-conciliator}: the first 2 to do commit-adopt, and a third round in which processors output their leader's value if possible. Each commit-adopt consists of 2 rounds as per~\Cref{algo:ca}.}
        \begin{center}
            \includegraphics[width=\textwidth]{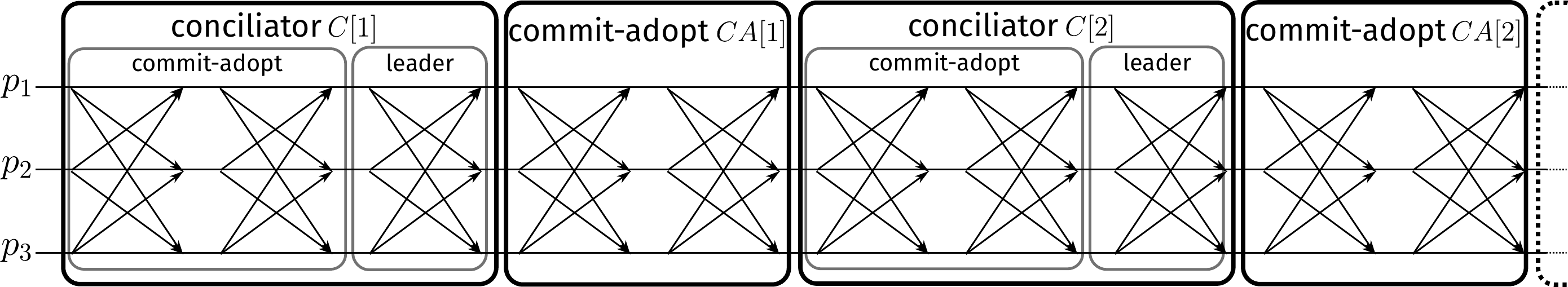}
        \end{center}
        \label{fig:proba-cons}
\end{figure}

\begin{theorem}
    \label{thm:proba-conciliator-correct}
    \Cref{algo:proba-conciliator} implements the conciliator specification with the probabilistic agreement guarantee.
\end{theorem}
\begin{proof}
    We must show the validity and termination properties of~\Cref{def:conciliator}, and the probabilistic agreement property of~\Cref{def:proba-agreement}.

    Termination is trivial: the algorithm terminates in 3 rounds.

    For Validity, if all start with $v$ then, by property of commit-adopt, all output $commit(v)$ from the commit-adopt at the end of round 2 and thus broadcast $commit(v)$, then all receive a strict majority of $commit(v)$ by the end of round 3, and thus all output $v$ according to~\Cref{algo:proba-output-1}.

    For the liveness property, it suffices to show that if all processors obtain the same well-behaved leader $l\in W_2$ in round 2, then all output the same value.
    So, suppose that all processors obtain the same well-behaved leader $l$ in round 2.
    First note that, since $l$ is a well-behaved processor, every processor receives $l$'s message and no processor outputs according to~\Cref{algo:proba-output-3}.
    Thus, it suffices to consider the following two cases:
    \begin{enumerate}[label=(\alph*)]
        \item $p_1$ and $p_2$ output according to~\Cref{algo:proba-output-1}.
        \item $p_1$ outputs according to~\Cref{algo:proba-output-1} and $p_2$ outputs according to~\Cref{algo:proba-output-2}.
    \end{enumerate}

    First, assume that both $p_1$ and $p_2$ output according to~\Cref{algo:proba-output-1} and assume by contradiction that $p_1$ outputs $v_1$ and $p_2$ outputs $v_2\neq v_1$.
    Then, $p_1$ received $commit(v_1)$ from a strict majority in round 3 and $p_2$ received $commit(v_2)$ from a strict majority in round 3.
    Since no processor broadcasts both $commit(v_1)$ and $commit(v2)$, with~\Cref{thm2}, we get a contradiction.

    Second, assume that $p_1$ outputs according to~\Cref{algo:proba-output-1} and $p_2$ outputs according to~\Cref{algo:proba-output-2}.
    Then $p_1$ receives $commit(v_1)$ from a strict majority in round 3.
    Thus, at least one well-behaved processor $p\in W_3$ broadcast $commit(v_1)$ in round 3, and thus all processors output either $commit(v_1)$ or $adopt(v_1)$ from the commit-adopt at the end of round 2.
    Thus, by the agreement property of commit-adopt~(\Cref{def:ca}), $p_2$ receives either $commit(v_1)$ or $adopt(v_1)$ from its leader and outputs~$v_1$.
\end{proof}

We are now ready to show that the we solve consensus in 20 rounds in expectation if we use the conciliator of~\Cref{algo:proba-conciliator}.
\begin{theorem}
    \label{thm:expected-rounds}
    In the generic consensus algorithm (\Cref{algo:gen-consensus}) instantiated with~\Cref{algo:ca} as commit-adopt and the probabilistic conciliator (\Cref{algo:proba-conciliator}) as conciliator, every processor produces a consensus output after 20 rounds of the IIAB model in expectation.
\end{theorem}
\begin{proof}
    By~\Cref{thm:proba-conciliator-correct} and~\Cref{def:proba-agreement}, each conciliator produces agreement with probability 1/2, and, if agreement is produced then the following commit-adopt is guaranteed to produce a consensus decision.
    Thus, in expectation, we reach a decision in $CA[2]$.
    This takes 2 conciliators (taking 3 no-equivocation rounds each) and 2 commit-adopts (taking 2 no-equivocation rounds each), for a total of 10 no-equivocation rounds or, equivalently, 20 IIAB rounds.
\end{proof}

\subsection{Consensus with deterministic safety and deterministic liveness}
\label{sec:dolev-strong}

In this section, we turn to solving consensus fully deterministically under an eventual stabilization assumption.

First, we present a deterministic conciliator algorithm.
This algorithm is parameterized by a number of rounds $N>1$ and guarantees agreement after $N+1$ rounds assuming the following deterministic-conciliation assumption:

\begin{definition}[Deterministic-conciliation assumption]
    \label{def:det-conciliation-assm}
    The deterministic-conciliation assumption holds when:
    \begin{enumerate}
        \item Participation is constant over the first $N$ rounds, i.e. $O_r=O_1$ for every round $r\leq N$, and
        \item The adversary is growing, i.e. for every round $r$, $F_r\subseteq F_{r+1}$, and
        \item Less than $N$ participants are impersonated in the first $N$ rounds, i.e.\ $\left|\bigcup_{r\leq N}F_r\right|<N$.
    \end{enumerate}
\end{definition}

\begin{definition}[Deterministic agreement property]
    \label{def:det-agreement}
    Assuming that the deterministic-conciliation assumption holds, every processor outputs the same value at the end of round $N+1$.
\end{definition}

Next, we obtain our fully-deterministic consensus algorithm by using the generic consensus algorithm (\Cref{algo:gen-consensus}) where we instantiate each conciliator $C[n]$ with the deterministic conciliator using parameter $N=2^n$ and, as before, we instantiate each commit-adopt $CA[n]$ with the algorithm of~\Cref{sec:ca}.

The fully-deterministic consensus algorithm guarantees deterministic safety and deterministic liveness under the following eventual stabilization assumption:
\begin{definition}[Eventual stabilization]
    An execution satisfies the eventual stabilization assumption when:
    \begin{enumerate}
            \item The adversary is growing, i.e. for every round $r$, $F_r\subseteq F_{r+1}$, and
            \item There is a round $R$, unknown to the processors and called the stabilization round, such that participation remains constant after round $R$ ($O_r=O_R$ for every round $r\geq R$).
    \end{enumerate}
\end{definition}

The idea is that, because we increase the number of rounds of the deterministic conciliator each time we run it, eventually we run an instance of the deterministic conciliator that starts after round $R$ and that runs for a number of rounds greater than there are failures; thus, by property of the deterministic conciliator, all processors output the same value and a decision is reached in the next commit-adopt.

Finally, by increasing the number of conciliation rounds exponentially, we ensure that the fully-deterministic consensus algorithm terminates in $O(\left|\bigcup_{r\geq R}F_r\right|)$ rounds, or equivalently $O(\left|O_R\right|)$ rounds, after round $R$.

\subsubsection{The deterministic conciliator}

The deterministic conciliator outputs after $N+1$ rounds, where $N$ is a parameter of the algorithm, and it must guarantee the validity property of conciliators and, assuming the deterministic-conciliation assumption above, it must guarantee agreement, i.e.\ that every processor output the same value.

To ensure agreement under the deterministic-conciliation assumption above, we simply use the Dolev-Strong algorithm~\cite{dolev_authenticated_1983} to reach agreement on the set of processors that participated in round 1 and on their inputs, and then deterministically select an output among the inputs.
Under the deterministic-conciliation assumption, the system essentially behaves like a traditional, fixed-participation synchronous system with message authentication and with less than $N$ malicious processors.
Thus the Dolev-Strong algorithm guarantees agreemnent after $N$ rounds.

There is only one small catch: we must guarantee that all processors output the same value, even the impersonated processors.
We easily achieve this by adding an extra round in which processors broadcast their candidate output and output the value received from a strict majority, if any, and otherwise output their own candidate output.

Finally, we must also guarantee the validity property even when the deterministic-conciliation assumption does not hold.
Fortunately, this is easily achieved by tweaking the deterministic rule used to select an output: we require each processor to output the value held by a strict majority of the processors if there is such a value.

We now give the full algorithm, noting that it is just the Dolev-Strong algorithm with a specific rule to pick an output in the last round and an additional round at the end.
\begin{algorithm}[Deterministic conciliator]
    \label{algo:dolev-strong}
    The algorithm is parameterize by a natural number $N>0$ and runs for $N$ rounds.
    Each processor $p$ maintains a local variable $e_p$, initialized to the empty set.
    \begin{itemize}
        \item In the send phase of round $1$, each online processor $p$ broadcasts the signed message $\tuple{p,1,\text{in}_p}$ where $\text{in}_p$ is $p$'s external input.
        \item In the receive phase of each round $1\leq r\leq N$, each online processor $p$ does the following: for each message of the form $\tuple{p_{r},r,\tuple{p_{r-1},r-1,\dots\tuple{p_1,1,v}}}$ that $p$ receives, if no processor repeats in the sequence and if $p_1$ does not appear in any tuple in the set $e_p$, then $p$ adds the tuple $(p_1,v)$ to the set~$e_p$.
        \item In the send phase of each round $1<r\leq N$, each online processor $p$ first initializes a local variable $m_p^r$ to the empty set.
            Then, for each message $m$ of the form $\tuple{p_{r-1},r-1,\tuple{p_{r-2},r-2,\dots\tuple{p_1,1,v}}}$ that $p$ received in the previous round $r-1$, if no processor repeats in the sequence then processor $p$ adds the signed message $\tuple{p,r,m}$ to $m_p^r$.
            Finally, processor $p$ broadcasts $m_p^r$.
        \item In the send phase of round $N+1$, each online processor $p$ determines a candidate output as follows and broadcasts it.
            \begin{itemize}
                \item If there is a value $v$ such that a strict majority of the processors appearing in a tuple in $e_p$ appear in a tuple whose second component is $v$, then $p$'s candidate output is $v$.
                \item Otherwise, $p$'s candidate output is the smallest value appearing in a tuple in $e_p$.
            \end{itemize}
        \item Finally, at the end of round $N+1$, each processor outputs as follows:
            \begin{itemize}
                \item If $p$ receives a value $v$ from a strict majority of the processors it hears of, then $p$ outputs $v$.
                \item Otherwise, $p$ outputs its own candidate output.
            \end{itemize}
    \end{itemize}
\end{algorithm}

\begin{lemma}
    \label{lem:dolev-strong}
    If the deterministic-conciliation assumption holds, then every processor $p$ that has not been impersonated has the same set $e_p$ at the end of round $N$.
\end{lemma}
\begin{proof}
    First, note that if a processor $p$ adds a tuple $\tuple{p',v}$ to its set $e_p$ in round $N$, then it must have received a signed message of the form $\tuple{p_N,N,\tuple{p_{N-1},N-1,\tuple{p_{N-2},N-2,\dots\tuple{p_1,1,v}}}}$, where no processor repeats in the sequence, in round $N$.
    Since $N>\left|\bigcup_{r\leq N}F_r\right|$, the sequence $p_1,p_2,\dots,p_N$ must contain a processor $p_j$ that is well-behaved in all rounds.
    Therefore, in round $j$, $p_j$ broadcasts $\tuple{p_j,j,\tuple{p_{j-1},j-1,\dots\tuple{p',1,v}}}$ in round $j$.
    Thus, every processor $p''$ adds the tuple $\tuple{p',v}$ to its set $e_{p''}$ by the end of round $j$ at the latest.

    Next, note that if a well-behaved processor $p$ adds a tuple $\tuple{p',v}$ to its set $e_p$ in a round $r<N$, then it must have received a message of the form $\tuple{p_r,r,\tuple{p_{r-1},r-1,\dots\tuple{p',1,v}}}$, where no processor repeats in the sequence, in round $r$.
    Therefore, $p$ broadcasts $\tuple{p,r+1,\tuple{p_r,r,\dots\tuple{p',1,v}}}$ in round $r+1$.
    Moreover, $p\notin \{p_r,p_{r-1},\dots,p'\}$ or $\tuple{p',v}$ would already be in $e_p$.
    Thus, every processor $p''$ adds $\tuple{p',v}$ to $e_{p''}$ by the end of round $r+1$ at the latest.

    We conclude that, at the end of round $N$, every processor $p$ that has not been impersonated has the same set $e_p$.
\end{proof}

We are now ready to show that the conciliator algorithm is correct.
\begin{theorem}
    \label{thm:dolev-strong}
    \Cref{algo:dolev-strong} implements a conciliator with the additional deterministic agreement property.
\end{theorem}
\begin{proof}
        We must show that~\Cref{algo:dolev-strong} satisfies the validity and termination property of conciliators (\Cref{def:conciliator}) and the deterministic agreement property (\Cref{def:det-agreement}).

        Termination is trivial: the algorithm terminates in $N+1$ rounds by definition.

        For Validity, note that, for every well-behaved processor $p'$, every processor $p$ adds $\tuple{p',in_{p'}}$ to the set $e_p$ at the end of round $1$.
        Moreover, note that for every two processors $p$ and $p'$, if $\tuple{p',v}\in e_p$ for some $v$, then $p'$ was online in round $1$.
        Thus, if all well-behaved processors have the same input $v$, then for every processor $p$, $v$ appears in strictly more than 1/2 of the tuples in $e_p$.
        Thus all well-behaved processors broadcast $v$ in round $N+1$ and all output $v$.

        Finally, for the deterministic agreement property, assume that the deterministic-conciliation assumption holds.
        Then, by~\Cref{lem:dolev-strong}, every processors $p$ that has not been impersonated has the same set $e_p$ at the end of round $N$, and thus all such processors broadcast the same candidate output in round $N+1$.
        Because the adversary is growing and participation is constant, there must be more processors that have not been impersonated than impersonated processors.
        Thus, every processor receives the same value $v$ from a strict majority of the processors it hears of in round $N+1$, and thus every processor outputs $v$.
\end{proof}

\subsubsection{Fully deterministic consensus}

\begin{theorem}
    \label{thm:det-consensus}
    Assume that the eventual stabilization assumption holds.
    Consider the generic consensus algorithm (\Cref{algo:gen-consensus}) instantiated with~\Cref{algo:ca} as commit-adopt and, for every $n$, with $C[n]$ being the deterministic conciliator of this section (\Cref{algo:dolev-strong}) with parameter $N=2^n$.
    The algorithm guarantees that every processor produces a consensus output in at most $O(\left|\bigcup_{r\geq R}F_r\right|)$ rounds after the stabilization round $R$.
\end{theorem}
\begin{proof}
    Let $b=\left\lceil\log_2\left(\left|\bigcup_{r\geq R}F_r\right|\right)\right\rceil+1$.
    If $b\geq R$, because $2^b>\left|\bigcup_{r\geq R}F_r\right|$, by~\Cref{thm:dolev-strong}, every processor produces the same output in the conciliator $C[b]$.
    Otherwise, each processor produces the same output in the first conciliator after round $R$.
    In any case, every processor produces a consensus decision at the end of the next commit-adopt.

    Note that commit-adopt takes a fixed 3 rounds and each conciliator $C[n]$ takes $2^n+1$ rounds.
    Thus, it takes at most $3b+\sum_{i=1}^b 2^i+1 = O(\left|\bigcup_{r\geq R}F_r\right|)$ to reach conciliator $C[b]$, and we are guarantee that all processors output after $O(\left|\bigcup_{r\geq R}F_r\right|)$ after round $R$.
\end{proof}

\section{Related Work}
\label{sec:related}

\remove{
\begin{itemize}
    \item Zikas and Gafni
    \item sleepy model of Pass and Shy~\cite{pass_sleepy_2017}; $N$ is known in this model; but is it necessary? They solve consensus with 1/2 corruptions using a longest-chain protocol like bitcoin (so, with probabilistic safety)
    \item 1/2 corruption, VRF consensus of Momose and Ren CCS2022~\cite{momose_constant_2022}; $N$ is fixed too but does not seem used in the algorithm. Needs eventually-stable participation for liveness.
    \item Malkhi 1/3 2022\cite{malkhi_byzantine_2022}
    \item Dahlia's blog posts: CA in the 1/2 case\cite{research_minority_2022}
    \item Algorand\cite{chen_algorand_2019}
    \item Afek and Gafni message adversary\cite{afek_simple_2015}
\end{itemize}
}

In order to simplify reasoning about distributed-computing problems, a number of earlier works use the notion of a message adversary that controls communication while processors themselves never fail.
Afek and Gafni~\cite{afek_simple_2015} consider a synchronous system with omissions, where the omissions of outgoing messages are preformed by a message adversary.
In this model a processor is never faulty and obeys the algorithm following the messages it receives.
Thus, when solving a task in the model, all processors are obliged to output, as opposed to only correct ones in failure-based models. 
Using the message adversary model, Afek and Gafni greatly simplify the proof of the asynchronous computability theorem.

We adopt a similar setting in the current paper: a processor never fails, but a malicious adversary might temporarily impersonate it and send messages on its behalf.
Santoro and Widmayer~\cite{santoro_agreement_2007} also study consensus under benign, transient communication failures, which resemble the message adversary of Afek and Gafni, and prove that consensus is impossible with just one benign failure per round.

In an unpublished manuscript written in 2019 (now published on arXiv~\cite{gafni_synchronyasynchrony_2023}), Gafni and Zikas introduce the mobile message adversary in the authenticated Byzantine, synchronous setting with fixed, known participation.
They show how to solve commit-adopt in the resulting model, and the commit-adopt algorithm we present (\Cref{algo:ca}) borrows the main algorithmic idea.
The introduction of unknown participation in the IIAB model however significantly complicates the matter.
Gafni and Zikas also present a flawed consensus algorithm for their model (the flaw is that their conciliator does not satisfy the Validity property of~\Cref{def:conciliator}).

Algorand~\cite{chen_algorand_2019} proceeds from round to round choosing an authenticated committee for the round.
A committee member does not a-priori know the other committee members, but nevertheless when getting a message from a committee member it can verify its eligibility to be on the committee.
Proof-of-Stake ensures, probabilistically, that the committee consists of a 1/3 fraction of honest members.
The main difference with out setting is that processors know in Algorand know, even if only in expectation, the cardinality of the committee.
This has profound algorithmic implications: in Algorand a processor can check with high probability whether some fraction $f$ of the committee supports some value, whereas in the IIAB model processors can only check whether a fraction $f$ of the processors they hear of support some value.

Pass and Shi~\cite{pass_sleepy_2017} propose the sleepy model, where processor wake up to participate, and may go back to sleep.
Thus like Algorand participation is not known a-priori, but unlike Algorand how many participate is unknown (even in expectation).
They consider minority corruption and employ a longest-chain mechanism achieving probabilistic safety.
The Snow-White paper~\cite{daian_snow_2019} extends the work of Pass and Shi to the Proof-of-Stake setting and proposes practical algorithms.

Sandglass~\cite{pu_safe_2022} achieves deterministic safety under a minority of benign failures even when the faulty set can grow from round to round.
The algorithm we present also works in the Sandglass model since there are no Byzantine failures.
Gorilla~\cite{pu_gorilla_2023} extends Sandglass to tolerate malicious failures.
Compared to the current work, Gorilla does not use message authentication.
Instead, Gorilla uses verifiable delay functions (VDF)~\cite{boneh_verifiable_2018} to achieve a form of proof-of-work with deterministic guarantees.

Monmose and Ren~\cite{momose_constant_2022} achieve deterministic safety in the sleepy model assuming a failure ratio of less than 1/2, but their algorithm requires eventually stable participation to achieve progress with probability 1.
Malkhi, Momose, and Ren~\cite{malkhi_byzantine_2022} remove the eventual stability requirement at the expense of reducing corruption to at most 1/3.
In following work, Malkhi, Momose, and Ren~\cite{malkhi_towards_2023} achieve deterministically safe consensus under a failure ratio of 1/2 without eventual stabilization requirements.
This matches the guarantees of the algorithm presented in~\Cref{sec:proba-consensus}.

In the probabilistic-safety domain, starting with Bitcoin's longest-chain protocol, a series of works solve consensus under dynamic-participation with probabilistic safety~\cite{pass_sleepy_2017,daian_snow_2019,badertscher_ouroboros_2018,damato_no_2022,goyal_instant_2021}.

\subsection{Conclusion}

Like all good basic research, the unpublished 2019 write-up of Gafni and Zikas (now published on arXiv~\cite{gafni_synchronyasynchrony_2023}), which was considered by all who knew about it elegant but esoteric, is found 3 years later to contain an algorithmic idea for commit-adopt that is relevant to permissionless systems. 

In this paper, we build on the work of Gafni and Zikas to solve consensus in a general model of permissionless systems called the IIAB model.
We solve consensus deterministically under a bounded adversary (which subsumes eventually stable participation), and probabilistically, even under an unbounded adversary, but relying on a leader-election oracle that can be implemented using verifiable random functions.

We present structured modular algorithms without trying to reduce message or round complexity.
This leads to remarkably simple algorithms ready to be taught in class tomorrow. 
In this we follow the time-honored tradition of separation of concerns: First investigate feasibility at the utmost level of clarity; only then optimize for efficiency.
The latter is left to future work.

Last but not least:
The remarkable FLP result~\cite{fischer_impossibility_1985} showed that consensus can be prevented if $n$ processors are required to proceed each missing one processor per round.
Santoro and Widmayer~\cite{santoro_agreement_2007} proved a similar impossibilty result in a synchronous system with a single processor suffereing message-omission failures each round, where the affected processor may change from round to round.

These two settings give rise to two separate notions of eventual synchrony (the latter better called eventual stationarity).
First, eventual in the sense of FLP, i.e.\ asynchronous then eventually synchronous.
Second, in the sense of Santoro and Widmayer, i.e.\ synchronous but mobile faults then eventually the faults become stationary (i.e.\ eventually at least one processor is not affected).

The former notion of eventual exhibits mathematical irregularity -- Authenticated Byzantine can synchronously solve consensus with minority faulty\footnote{Some claim that the Dolev-Strong algorithm achieves 99\% fault tolerance, but the algorithm cannot satisfy the validity propety of consensus if at least half are faulty; moreover, if at least half are faulty, achieving consensus is not of much use because there is no way for an external observer to tell what the non-fauty processes decided.}, but the eventual version requires a bound of one third.
There is no such gap from half to one third in the eventually stationary notion:
As we have shown, the IIAB model solves consensus under eventually stationary for minority corruptions.

\printbibliography

\end{document}